%% file: InformationAssimilation.tex
\documentclass[journal]{IEEEtran}

\usepackage{amsfonts,amssymb,mathtools,amsthm,bm}
\usepackage[OT1]{fontenc}
\usepackage[numbers]{natbib}
\usepackage[inline]{enumitem}
\usepackage[colorlinks,citecolor=blue,urlcolor=blue]{hyperref}
\usepackage{color}
\usepackage{comment}

\input{localDefs.tex}

\title{Bayesian data assimilation based on a family of outer measures}
\author{Jeremie Houssineau and Daniel E.\ Clark
\thanks{J.\ Houssineau is with the Department of Statistics and Applied Probability, National University of Singapore.}
\thanks{D.E.\ Clark is with the School of Engineering and Physical Sciences, Heriot-Watt University, Edinburgh.}}

\begin{document}

\maketitle


\begin{abstract}
A flexible representation of uncertainty that remains within the standard framework of probabilistic measure theory is presented along with a study of its properties. This representation relies on a specific type of outer measure that is based on the measure of a supremum, hence combining additive and highly sub-additive components. It is shown that this type of outer measure enables the introduction of intuitive concepts such as pullback and general data assimilation operations.
\end{abstract}



\section*{Introduction}

Uncertainty can be considered as being inherent to all types of knowledge about any given physical system and needs to be quantified in order to obtain meaningful characteristics for the considered system. Uncertainty originates from randomness as well as from lack of precise knowledge and is usually modelled by a probability measure on a state space on which the system of interest can be represented. However, in situations where the lack of knowledge is acute, probability measures are not always suitable for representing the associated uncertainty, and different approaches have been introduced based on \emph{fuzzy sets} \cite{Zadeh1965}, \emph{upper and lower probabilities} \cite{Dempster1967} or \emph{possibilities} and \emph{plausibilities} \cite{Shafer1976}. Various extensions of the concept of plausibility have also been studied \citep{Yen1990,Friedman2001}.

In this article, we study a representation of uncertainty that is completely within the standard framework of probabilistic measure theory. We do not claim that the approach considered here is more general than any previously introduced representation of uncertainty, it is rather the considered framework that will prove to be beneficial when handling uncertainty. Indeed, the proposed approach allows for deriving a Bayesian data assimilation operation that combines independent sources of information. It is shown that this operation yields Dempster's rule of combination \citep{Dempster1968,Shafer1986} when the proposed representation of uncertainty reduces to a plausibility, which furthers the relation between the Bayesian and Dempster-Shafer approaches discussed in \cite{Brodzik2009}.
The mathematical properties of the proposed Bayesian data assimilation operation are also studied, similarly to \cite{Brodzik2009}. This operation can be used in conjunction with random variables describing different types of physical systems, for instance it has been used in \cite[Chapt.~2]{Houssineau2015_Thesis} together with a general representation of stochastic populations \cite{Houssineau2016} in order to derive a principled solution to the \emph{data association} problem which underlies data assimilation for probabilistic multi-object dynamical systems and which is usually handled in an ad-hoc way. This enabled new methods for multi-target filtering to be derived, as in \cite{Delande2016} and \cite[Chapt.~4]{Houssineau2015_Thesis}.

The proposed representation of uncertainty also enables the definition of a meaningful notion of \emph{pullback} -- as the inverse of the usual concept of pushforward measure -- that does not require a modification of the measurable space on which the associated function is defined, i.e.\ it does not require the considered measurable function to be made bi-measurable.

The structure of the paper is as follows. Examples of limitations encountered when using probability measures to represent uncertainty are given in Section~\ref{sec:uncertaintyAndProbabilityMeasure}. The concepts of outer measure and of convolution of measures are recalled in Section~\ref{sec:fundamentalConcepts}, followed by the introduction of the proposed representation of uncertainty in Section~\ref{sec:probabilisticConstraint} and examples of use in Section~\ref{sec:exampleProbaConstraint}. Preliminary results about pushforward and pullback measures are given in Section~\ref{sec:operationsOnMeasureConstraints} before the statement of the main theorems in Section~\ref{sec:fusion}.

\subsection*{Notations}

The sets of non-negative integers is denoted $\bbN$, and the set of real (non-negative) numbers is denoted $\bbR$ ($\bbR^+$). The power set of a set $A$ is denoted $\wp(A)$, where ``$\wp$'' is referred to as the ``Weierstrass p''. For any set $\boE$, the function on $\boE$ which is everywhere equal to one is denoted $\one$. For any set $\boF$, any $f,f' : \boE \to \boF$ and any $g,g':\boE \to \bbR$, we define the mappings $g \cdot g'$ and $f \ltimes f'$ on $\boE$ as
\eqnsa{
g \cdot g' & : x \mapsto g(x)g'(x) \in \bbR,\\
f \ltimes f' & : x \mapsto (f(x),f'(x)) \in \boF \times \boF,
}
and the mappings $g \rtimes g'$ and $f \times f'$ on $\boE \times \boE$ as
\eqnsa{
g \rtimes g' & : (x,x') \mapsto g(x)g'(x') \in \bbR,\\
f \times f' & : (x,x') \mapsto (f(x),f'(x')) \in \boF \times \boF.
}

We consider two complete probability spaces $(\Omega,\Sigma,\bbP)$ and $(\Omega',\Sigma',\bbP')$ that are assumed to represent independent sources of uncertainty. If two random variables $X$ and $X'$ are defined on $\Omega$, then the associated joint random variable is $X \ltimes X'$ and its law is found to be the pushforward $(X \ltimes X')_*\bbP$, moreover, if $X$ and $X'$ are real random variables then their product is $X \cdot X'$. If $X'$ is defined on $\Omega'$ then the corresponding joint random variable will be $X \times X'$ and its law will be $(X \times X')_*(\bbP \rtimes \bbP')$.

If $(\boE,\calE)$ is a measurable space then $\boM(\boE)$ and $\boM_1(\boE)$ denote the set of measures and the set of probability measures on $\boE$ respectively. If $\boX$ is a Polish space, then $\calB(\boX)$ denote the Borel $\sigma$-algebra on $\boX$.

\section{Uncertainty and probability measures}
\label{sec:uncertaintyAndProbabilityMeasure}

Random experiments are mathematically defined by probability spaces, defining possible events and attributing them a probability. However, the use of a probability space for representing uncertainty that is not induced by pure randomness, i.e.\ for describing random or non-random experiments about which only partial knowledge is available, can become less intuitive as well as technically challenging in some cases. 
In the following examples, we study when uncertainty can be encoded into a probability measure as well as cases where the existence of a procedure for doing so becomes less clear.

\begin{example}
Let $X$ be a random variable from $(\Omega,\Sigma,\bbP)$ to $(\bbR,\calB(\bbR))$ and assume that it is only known that $X$ has its image in the Borel subset $A$ in $\calB(\bbR)$ with probability $\alpha \in [0,1]$. This information can be encoded via the sub-$\sigma$-algebra $\calA \defeq \{\emptyset, A, A^c, \bbR\}$ of $\bbR$ by defining the law $p$ of $X$ on $(\bbR,\calA)$ rather than on $(\bbR,\calB(\bbR))$ and by setting $p(A) = \alpha$. Similarly, if nothing is known about $X$, then this can be encoded via the \emph{trivial} sub-$\sigma$-algebra $\{\emptyset, \bbR\}$.
\end{example}

The concept of sub-$\sigma$-algebra is useful for describing different levels of knowledge, such as when used for conditional expectations \citep[Chapt.~27]{Loeve1978}. However, we will see in the next example that their use can become challenging in some situations.

\begin{example}
Let $p$ be a probability measure on $(\boE,\calE)$ and let $p'$ be another probability measure on $(\boE,\calE')$, with $\calE' \subset \calE$, then for any scalar $a \in (0,1)$, the probability measure ${q_a = (1-a) p + a p'}$ can only be defined on the coarsest $\sigma$-algebra, that is $\calE'$. When considering the extreme case where $\calE'$ is the trivial $\sigma$-algebra $\{\emptyset,\boE\}$, it results that nothing is known about $q_a$, however small is~$a$. One way to bypass this drawback is to single out a finite reference measure $\lambda$ in $\boM(\boE)$ and to define an extended version of $p'$ denoted $\bar{p}'$ as a \emph{uniform} probability measure as follows
\eqns{
\frall{\forall A \in \calE}  \bar{p}'(A) \defeq \dfrac{\lambda(A)}{\lambda(\boE)}.
}
In this way, the probability of a given event in $\calE$ is equal to the probability of any other event of the same ``size'' with respect to the measure\ $\lambda$. In other words, no area of the space $\boE$ is preferred over any other. Besides the facts that a reference measure is required and that the size of the space is limited (there is no uniform measure over the whole real line with respect to the Lebesgue measure), this way of modelling the information is not completely equivalent to the absence of information. There exist ways of modelling uncertainty on a probability measure itself, such as with Dirichlet processes \citep{Ferguson1973} and to some extent with Wishart distributions \citep{Wishart1928}; yet, these solutions do not directly help with the non-informative case since they require additional parameters to be set up.
\end{example}

\begin{example}
\label{ex:HMM}
Consider a discrete-time joint process $(X_n,Y_n)_{n\geq 0}$ assumed to be a hidden Markov model \cite{Baum1966} and consider the case where realisations $y_n$ of the observation process are not received directly but are instead known to be in some subset $A_n$ of the observation space. This is in fact what happens in practice since most real sensors are finite-resolution, e.g.\ if the sensor is a camera then $A_n$ would be one of the pixels in the received image at time step $n$. This type of data could be treated as such using non-linear filtering methods, as in \cite{Houssineau2015} for finite-resolution radars; however the associated uncertainty is often modelled by a Gaussian distribution so that the Kalman filter can be used instead. Yet, replacing information of the form $y_n \in A_n$ by a probability distribution can have a great impact on the value of the denominator in Bayes' theorem. Although this is unimportant in Kalman filtering since the numerator is equally affected, the value of the denominator does matter in different contexts, such as in multi-target tracking \cite{Blackman1986} where it is used as a weighting coefficient for the relative importance of a given \emph{track} when compared to others \cite{Pace2013}. We will see in Section~\ref{sec:fusion} that the Gaussian distribution can be replaced by another object that better represents data of the form $y_n \in A_n$ while preserving the advantages of Kalman filtering.
\end{example}

Another important aspect of probability theory is the gap between probability measures on countable and uncountable sets, as explained in the following example.

\begin{example}
Let $\boX$ be a countable set equipped with its discrete $\sigma$-algebra and assume that some physical system can be uniquely characterised by its state in $\boX$. Let $X$ be a random variable on $(\Omega,\Sigma,\bbP)$ and $X'$ be another random variable on another probability space $(\Omega',\Sigma',\bbP')$ and assume that both $X$ and $X'$ represent some uncertainty about the same physical system. The fact that the two random variables represent the same system can be formalised by the event $\Delta \defeq \{(x,x) \st x \in \boX\}$, which is the diagonal of $\boX \times \boX$. The information contained in the laws $p \defeq X_*\bbP$ and $p' \defeq (X')_* \bbP'$ can then be \emph{combined} into a conditional probability measure $\hat{p}(\cdot \given \Delta) \in \boM_1(\boX)$, characterised by
\eqns{
\frall{\forall B \subseteq \boX} \hat{p}(B \given \Delta) \defeq \dfrac{p \rtimes p'(B \times B \cap \Delta)}{p \rtimes p'(\Delta)},
}
where $p$ and $p'$ are assumed to be \emph{compatible}, i.e.\ that $p \rtimes p'(\Delta) \neq 0$. This result is justified by the transformation of the conditional probability measure $p\rtimes p'( \cdot \given \Delta)$ with support on $\Delta$ into a probability measure on $\boX$ (isomorphism $\Mod 0$). Let $w,w' : \boX \to [0,1]$ be the probability mass functions induced by $p$ and $p'$ and characterised by
\eqns{
p = \sum_{x \in \boX} w(x) \delta_{x}, \AND p' = \sum_{x \in \boX} w'(x) \delta_{x},
}
then the probability measure $\hat{p}(\cdot \given \Delta)$ can be more naturally characterised via its probability mass function $\hat{w}$ on $\boX$, which is found to be
\eqns{
\hat{w} : x \mapsto \dfrac{w(x) w'(x)}{\sum_{y \in \boX} w(y) w'(y)},
}
However, if $\boX$ is uncountable and equipped with its Borel $\sigma$-algebra $\calB(\boX)$ and if the probability measures $p$ and $p'$ are diffuse on $\boX$, then they will not be compatible by construction. Indeed, even though the diagonal $\Delta$ can still be defined and is measurable under extremely weak conditions on $\boX$, it holds that $p \rtimes p'(\Delta) = 0$ (more specifically, the diagonal $\Delta$ of $\boX\times\boX$ is measurable in a separable metric space \cite[Lemma~6.4.2]{Bogachev2007}, and remains measurable if $\boX$ is generalised to a Hausdorff topological space with a countable base; an interesting result from \cite{Nedoma1957}, detailed in \cite[Chapt.~21]{Schechter1996}, is that the diagonal $\Delta$ is never measurable when the cardinality of $\boX$ is strictly larger than the cardinality of the continuum.) The fact that $p \rtimes p'(\Delta) = 0$ is caused by the strong assumption that the probabilities $p(B)$ and $p'(B)$ are known for all measurable subsets in $\calB(\boX)$ and, because $p$ and $p'$ are diffuse, tend to zero when $B$ reduces to a singleton. The introduction of an appropriately coarse sub-$\sigma$-algebra on $\boX$, such as the one generated by a given countable partition, would allow for recovering some of the results that hold for countable spaces. However, such an approach will not be natural or intuitive in most of the situations. Alternatively, if $p$ and $p'$ are absolutely continuous with respect to a reference measure $\lambda \in \boM(\boX)$, then the probability density $\hat{f}$ of $\hat{p}(\cdot \given \Delta)$ can be expressed as a function of the probability densities $f$ and $f'$ of $p$ and $p'$ respectively as
\eqns{
\hat{f} : x \mapsto \dfrac{f(x)f'(x)}{\int f(y)f'(y) \lambda(\d y)}.
}
However, to resort to this approach amounts to ignoring the incompatibility between the two random variables, in addition to the loss of meaningful interpretation of the denominator which becomes dependent on the choice of the reference measure. In a filtering context, one of the laws, say $p$, would be the prior whereas $p'$ would represent the observation. Although the prior might be appropriately represented by a probability measure, the uncertainty in the observation is often mostly due to lack of knowledge for which a probability measure might be ill-suited.
\end{example}

\begin{example}
Attributing a probability measure to data originated from natural language is often inappropriate. For instance, if an observer locates an object in the real world \emph{around} position $x \in \bbR^3$, then representing this data with a distribution on $\bbR^3$ centred on $x$, such as a Gaussian distribution with a given variance, highly overstates the given information. The observer is not giving the exact probability for the object to be in any measurable subset of $\bbR^3$; he only gives the fact that the probability should be low for subsets that are far from $x$ in a given sense and high when close to or when containing~$x$. 
\end{example}

Overall, there is a need for the introduction of additional concepts that could account for these non-informative types of knowledge, and which would in turn enable data assimilation to be performed in more general spaces. The objective in the next section is to find an alternative way of representing uncertainty while staying in the standard formalism of measure and probability theory.

\section{Fundamental concepts}
\label{sec:fundamentalConcepts}

We first recall two concepts of measure theory that will be useful in this article, namely the concepts of outer measure and of convolution of measures.

\subsection{Outer measure}
\label{sec:outerMeasure}

The concept of \emph{outer measure} is fundamental in measure theory and is defined as follows.

\begin{definition}
An outer measure on $\boE$ is a function $\mu : \wp(\boE) \to [0,\infty]$ verifying the following conditions:
\begin{enumerate}[label=\itshape\alph*\upshape)]
\item (Null empty set) $\mu(\emptyset) = 0$
\item (Monotonicity) if $A \subseteq B \subseteq \boE$ then $\mu(A) \leq \mu(B)$
\item (Countable sub-additivity) for every sequence $(A_n)_{n \in \bbN}$ of subsets of $\boE$
\eqns{
\mu\bigg(\bigcup_{n \in \bbN} A_n\bigg) \leq \sum_{n \in \bbN} \mu(A_n).
}
\end{enumerate}
\end{definition}

Outer measures allow for constructing both $\sigma$-algebras and measures on them via Carath\'eodory's method \cite[Sect.~113]{Fremlin2000}: the outer measure $\mu$ induces a $\sigma$-algebra $\calX$ of subsets of $\boX$ composed of sets $A$ verifying
\eqns{
\frall{\forall C \subseteq \boX} \mu(C) = \mu(C \cap A) + \mu(C \cap A^c),
}
which are referred to as \emph{$\mu$-measurable sets}. The measure space $(\boX,\calX,\mu)$ is a complete measure space. A classical example of measures constructed in this way is the Lebesgue measure on the $\sigma$-algebra of Lebesgue measurable subsets of $\bbR$. If $\boF$ is a set, $\mu$ and $\mu'$ are outer measures on $\boE$ and $\boF$ respectively and $f : \boE \to \boF$ is a function, then
\eqns{
C \subseteq \boF \mapsto \mu(f^{-1}[C]) \AND A \subseteq \boE \mapsto \mu'(f[A])
}
are outer measures. One way of constructing outer measures from measures is as follows: let $(\boX,\calX,m)$ be a measure space, then $m$ induces an outer measure $m^*$ on $\boX$ defined as
\eqnl{eq:prop:probabilisticConstraint:inducedOuterMeas}{
m^*(C) = \inf\{ \mu(A) \st A \in \calX \ET A \supseteq C \},
}
for any $C \subseteq \boX$.

\subsection{Convolution of measures}
\label{ssec:convolution}

An operation that will prove to be of importance is the operation of convolution of measures. We consider that the set $\boX$ is a Polish space equipped with its Borel $\sigma$-algebra $\calB(\boX)$. 

\begin{definition}
\label{def:convolution}
Let $m$ and $m'$ be two finite measures on a topological semigroup $(\boX,\cdot)$, then the convolution $m * m'$ of $m$ and $m'$ is defined as
\eqnl{eq:def:convolution}{
m * m' (A) = m \rtimes m'(\{ (x,y) \st x \cdot y \in A \}),
}
for any $A \in \calB(\boX)$.
\end{definition}

The set function $m * m'$ defined in \eqref{eq:def:convolution} is a measure on $(\boX,\calB(\boX))$. The following properties of the operation of convolution are corollaries of \cite[Sect.~444A-D]{Fremlin2000}.

\begin{corollary}
If $(\boX,\cdot)$ is a topological semigroup, then it holds that
\eqns{
m * (m' * m'') = (m * m') * m''
}
for all finite measures $m$, $m'$ and $m''$ on $\boX$. If $(\boX,\cdot)$ is commutative, then it holds that $m * m' = m' * m$ for all finite measures $m$ and $m'$ on $\boX$.
\end{corollary}

\section{Measure constraint}
\label{sec:probabilisticConstraint}

Henceforth, we consider a space $\boX$ which is assumed to be a Polish space. We denote $\boL^0(\boX,\bbR)$ the set of all measurable functions from $(\boX,\calB(\boX))$ to $(\bbR,\calB(\bbR))$ and we use the shorthand notation $\boL^0(\boX) \defeq \boL^0(\boX,\bbR^+)$ for the subset of $\boL^0(\boX,\bbR)$ made of non-negative functions. The general definition of measure constraint is given below, followed by the introduction of more specific properties and operations.

\subsection{Definition of measure constraint}

Using the notion of outer measure defined in the previous section as well as the technical results about the set $\boL^0(\boX,\bbR)$ detailed in \cite[Sect.~245]{Fremlin2000}, we introduce the concept of \emph{measure constraint} as follows.

\begin{definition}
\label{def:constraint}
Let $M$ be a measure on $\boL^0(\boX,\bbR)$, if it holds that the function~$\mu_M$ defined on the power set $\wp(\boX)$ of $\boX$ as
\eqnl{eq:outerMeasureConstraint}{
\mu_M : A \mapsto M(\chi(A,\cdot))
}
is an outer measure for a given collection of measurable functions $\{\chi(A,\cdot)\}_{A \subseteq \boX}$ on $\boL^0(\boX,\bbR)$, then $M$ is said to be a \emph{measure constraint} on $\boX$ with \emph{characteristic function}~$\chi$. If $m$ is a finite measure on $\boX$ verifying
\eqnl{eq:def:constraint}{
m(B) \leq M(\chi(B,\cdot))
}
for any $B \in \calB(\boX)$ and
\eqnl{eq:def:constraint:cond}{
m(\boX) = M(\chi(\boX,\cdot)),
}
then $M$ is said to be dominating $m$.
\end{definition}

The motivation behind the introduction of measure constraints is to partially describe a measure by limiting the mass in some areas while possibly leaving it unconstrained elsewhere. Measure constraints that would bound a measure from below could also be defined using the associated concept of \emph{inner measure}. In general, a measure on $\boX$ could be dominated by a measure on the set $\boL^0(\boY,\bbR)$ of measurable functions on a different set $\boY$, as long as the associated characteristic function $\chi$ is defined accordingly.

\begin{remark}
In Definition~\ref{def:constraint}, the condition that $\mu_M$ is an outer measure is used to reduce the set of measures on $\boL^0(\boX,\bbR)$ that would verify \eqref{eq:def:constraint} to the ones that have natural properties. As explained by \cite[Sect.~113B]{Fremlin2000}: ``The idea of the \emph{outer} measure of a set $A$ is that it should be some kind of upper bound for the possible measure of $A$''. In fact, the use of outer measures as a way of dealing with uncertainty has first been proposed by \cite{Fagin1990}. In particular, the condition of monotonocity imposes that if a given mass is allowed in a set $A$ then at least the same mass should be allowed in a larger set $B \supseteq A$. Similarly, the condition of sub-additivity allows for reaching the maximum mass $m(\boX)$ in several disjoint sets while still verifying $M(\chi(\boX,\cdot)) = m(\boX)$.
\end{remark}

\begin{remark}
A direct consequence of \eqref{eq:def:constraint} is that a lower bound for $m(B)$ is also available for any $B \in \calB(\boX)$ and is found to be
\eqnl{eq:inducedLowerBound}{
m(B) = m(\boX) - m(B^c) \geq m(\boX) - M(\chi(B^c, \cdot)).
}
The information provided by this lower bound is limited since $\mu_M$ is sub-additive and might reach $m(\boX)$ on any given set $B' \subset \boX$. In this case, \eqref{eq:inducedLowerBound} only implies that $m(B') \geq 0$ which is not informative.
\end{remark}

We will be particularly interested in the situation where $m$ is a probability measure, in which case $M$ satisfies $M(\chi(\boX,\cdot)) = m(\boX) = 1$ and is said to be a \emph{probabilistic constraint}. The advantage with condition \eqref{eq:def:constraint:cond} is that if $m$ is a finite measure that is not a probability measure, then $m$ and $M$ can be renormalised to be respectively a probability measure and a probabilistic constraint by dividing the inequality \eqref{eq:def:constraint} by the total mass $m(\boX)$. When uncertainty is induced by a lack of knowledge on the actual law of a given random experiment then this law should be dominated by the considered probabilistic constraint. However, in cases of uncertain but non-random experiments, there is no such thing as a ``true'' law and the probabilistic constraint becomes the primary representation of uncertainty.

A useful case is found when $M$ is supported by the set $\boL^{\infty}(\boX)$ of non-negative bounded measurable functions and, assuming that $\sup : \boL^{\infty}(\boX) \to \bbR$ is measurable, when $\chi$ is such that
\eqnl{eq:characteristicSimple}{
\chi : (A,f) \mapsto \| \ind{A} \cdot f \|,
}
where $\ind{A}$ denotes the indicator of $A$ and where $\| \cdot \|$ is the uniform norm on $\boL^{\infty}(\boX)$. All outer measures do not take the form assumed in \eqref{eq:outerMeasureConstraint} with $\chi$ as in \eqref{eq:characteristicSimple}, but this case offers suitably varied configurations by combining a linear part and a very sub-additive part, that is the measure by $M$ and the uniform norm respectively. This case will be understood as the default situation in the sense that \eqref{eq:characteristicSimple} will be considered when the characteristic function of a measure constraint is not specified. The subset of probabilistic constraints with such a supremum-based characteristic function is denoted $\boC_1(\boX)$ and we consider the weighted semi-norm $\| \cdot \|$ defined on this set as
\eqns{
\| \cdot \| : M \mapsto \int \| f \| M(\d f).
}
This semi-norm is not a norm since $\| M \| = 0$ does not imply that $M$ is the null measure in general.

\subsection{Properties of measure constraints}

Henceforth, $\calL^{\infty}(\boX)$ will denote the Borel $\sigma$-algebra induced on $\boL^{\infty}(\boX)$ by the topology of convergence in measure on $\boL^0(\boX)$ presented in \cite[Sect.~245]{Fremlin2000}. The characterisation of $\calL^{\infty}(\boX)$ is highly technical and out of the scope of this article so that the measurability of the function $f \mapsto \sup f$ and the mapping $f \mapsto f \circ \xi$, for some measurable mapping $\xi$ in $\boX$, is assumed rather than demonstrated.

Given the definition of $\boC_1(\boX)$, it appears that many different measures in this set will dominate exactly the same probability measures in $\boM_1(\boX)$ since a rescaling of the functions in the support can be compensated for by a rescaling of the measure itself. However such a multiplicity can be avoided by identifying a canonical form and a rescaling procedure as in the following proposition.

\begin{proposition}
\label{prop:probaConstraintGeneralRemark:rescaling}
For any probabilistic constraint $M \in \boC_1(\boX)$, there exists a probabilistic constraint $M^{\scl} \in \boC_1(\boX)$ that is equivalent to $M$ and is also a probability measure, which can be determined via the following rescaling procedure:
\eqnl{eq:rescaleConstraint}{
\frall{\forall F \in \calL^{\infty}(\boX)} M^{\scl}(F) \defeq \int \ind{F}(f^{\scl}) \Ssup{f} M(\d f),
}
where the function $f^{\scl} \in \boL^{\infty}(\boX)$ is defined as
\eqns{
f^{\scl} \defeq 
\begin{cases*}
\dfrac{f}{\Ssup{f}} & if $\Ssup{f} \neq 0$ \\
\one & otherwise.
\end{cases*}
}
\end{proposition}

\begin{proof}
By construction, it holds that the support of $M^{\scl}$ is included in the subset of $\boL^{\infty}(\boX)$ made of functions with uniform norm equal to one, so that $M^{\scl}(\boL^{\infty}(\boX)) = \| M^{\scl} \| = 1$. The measure $M^{\scl}$ is then both a probability measure and a probabilistic constraint. We now have to show that $M$ and $M^{\dagger}$ dominate the same probability measures: Let $p \in \boM_1(\boX)$ be a probability measure dominated by $M$, then
\eqns{
m(B) \leq \int \| \ind{B} \cdot f \| M(\d f) = \int \| \ind{B} \cdot f^{\scl} \| M'(\d f),
}
for all $B \in \calB(\boX)$, where $M'(\d f) \inteq \| f \| M(\d f)$ so that, by a change of variable,
\eqns{
m(B) \leq \int \| \ind{B} \cdot f^{\scl} \| M'(\d f) = \int \| \ind{B} \cdot f' \| M^{\scl}(\d f'),
}
which terminates the proof.
\end{proof}

\begin{remark}
The definition of $f^{\scl}$ when $\| f \| = 0$ is irrelevant because of the form of \eqref{eq:rescaleConstraint}. Yet, considering $f^{\scl} = \one$ when $\| f \|= 0$ implies that $f^{\scl}$ is always in the subset
\eqns{
\boL(\boX) \defeq \{f \in \boL^{\infty}(\boX) \st \|f\| = 1\}
}
of measurable functions with uniform norm equal to one. Probabilistic constraints in the set $\boC^*_1(\boX) \defeq \boM_1(\boL(\boX))$ will be referred to as \emph{canonical probabilistic constraints}.
\end{remark}

The unary operation of Proposition~\ref{prop:probaConstraintGeneralRemark:rescaling} does not affect the norm, i.e.\ the equality $\|M^{\scl}\| = \| M \|$ holds by construction for any $M$ in $\boC_1(\boX)$. It is also \emph{idempotent} and it distributes over the product $\rtimes$, that is $(M^{\scl})^{\scl} = M^{\scl}$ and $(M \rtimes M')^{\scl} = M^{\scl} \rtimes M'^{\scl}$ hold for any $M, M' \in \boC_1(\boX)$. Moreover, any canonical probabilistic constraint $P \in \boC^*_1(\boX)$ verifies $P^{\scl} = P$.

\begin{remark}
If $M \in \boC_1(\boX)$ is a probabilistic constraint and if $M^{\scl}$ is supported by the set $\boI(\boX)$ of measurable indicator functions, then the conditions \eqref{eq:def:constraint} and \eqref{eq:def:constraint:cond} can be replaced by ``$m$ agrees with the measure induced by $\mu_M$'', i.e.\ $m(B) = \mu_M(B)$ for any $B$ in the $\sigma$-algebra of $\mu_M$-measurable subsets.
\end{remark}

\begin{proposition}
\label{prop:probaConstraintGeneralRemark:joint} If $X,X' \in \boL^0(\Omega,\boX)$ are two independent random variables on $\boX$ with respective laws $p$ and $p'$ and if $M$ and $M'$ are probabilistic constraints for $p$ and $p'$ respectively, then the joint law $p \rtimes p' \in \boM_1(\boX)$ verifies
\eqns{
p \rtimes p'(\hat{B}) \leq \int \| \ind{\hat{B}} \cdot (f \rtimes f') \| M(\d f) M'(\d f')
}
for any $\hat{B} \in \calB(\boX) \otimes \calB(\boX)$.
\end{proposition}

Proposition~\ref{prop:probaConstraintGeneralRemark:joint} is only an implication since a joint probability measure dominated by a probabilistic constraint of the same form might not correspond to independent random variables. This means that another concept is needed to describe this special class of joint probabilistic constraints.

\begin{definition}
\label{def:independentlyConstrained}
Let $X \in \boL^0(\Omega,\boX)$ and $X' \in \boL^0(\Omega',\boX)$ be two random variables 
then $X$ and $X'$ are said to be \emph{independently constrained} by $\hat{M} \in \boC_1(\boX \times \boX)$ if there exist $M, M' \in \boC_1(\boX)$ such that $\hat{M}$ can be expressed as a convolution on the semigroup $(\boL^{\infty}(\boX),\rtimes)$ as
\eqnl{eq:def:independentlyConstrained}{
\hat{M} = M * M',
}
recalling that, in this case,
\eqns{
M * M' (F) \defeq \int \ind{F}(f \rtimes f') M(\d f) M'(\d f')
}
for any measurable subset $F$ of $\boL^{\infty}(\boX \times \boX)$.
\end{definition}

Definition~\ref{def:independentlyConstrained} implies that any function $\hat{f}$ in the support of $\hat{M}$ is such that there exist $f$ and $f'$ in $\boL^{\infty}(\boX)$ for which $\hat{f} = f \rtimes f'$. Functions of this type can be said to be \emph{separable}. The concept of independently constrained random variables introduced in Definition~\ref{def:independentlyConstrained} differs in general from the standard concept of independence since, even if $X$ and $X'$ were defined on the same probability space, then
\begin{enumerate*}[label=\itshape\alph*\upshape)]
\item independently constrained random variables might actually be correlated, and
\item independent random variables that are not known to be independent might be represented by a probabilistic constraint that does not exclude correlation.
\end{enumerate*}
However, the concepts coincide when the involved probabilistic constraints are equivalent to probability measures. In cases where the probability spaces on which these random variables are defined are different, independence or correlations cannot even be defined.

In order to illustrate the use of measure constraints for modelling uncertainty, several examples of probabilistic constraints are given in the next section.

\section{Examples of probabilistic constraints}
\label{sec:exampleProbaConstraint}

Throughout this section, $X \in \boL^0(\Omega,\boX)$ will denote a random variable and $P \in \boC^*_1(\boX)$ will be assumed to dominate its law $p \defeq X_*\bbP$.

\subsection{Uninformative case}

If $P = \delta_{\one}$ then the only information provided by $P$ on $X$ is that
\eqns{
p(B) \leq \int \| \ind{B} \cdot f \| P(\d f) = 1
}
for any $B \in \calB(\boX)$, which is non-informative since $p$ is already known to be a probability measure. In other words, nothing is known about the random variable $X$. For instance, the $\sigma$-algebra of $\mu_P$-measurable sets is the trivial $\sigma$-algebra $\{ \emptyset, \boX \}$.

\subsection{Indicator function}

If there exists $A \in \calB(\boX)$ such that $P = \delta_{\ind{A}}$, then it holds that
\eqns{
p(B) \leq \| \ind{A \cap B} \| =
\begin{cases*}
1 & if $A \cap B \neq \emptyset$ \\
0 & otherwise,
\end{cases*}
}
for any $B \in \calB(\boX)$. As a probability measure, $p$ is always less or equal to $1$, so that the only informative part in the previous inequality is:
\eqns{
\frall{\forall B \in \calB(\boX)} ( A \cap B = \emptyset ) \implies ( p(B) = 0 ),
}
that is, all the probability mass of $B$ lies within $A$. This means that the random variable $X$ is only known to be in $A$ almost surely. A uniform distribution over $A$, assuming it can be defined, would not model the same type of knowledge as the corresponding interpretation would be ``realisations of $X$ are equally likely everywhere inside $A$'', whereas the interpretation associated with the probabilistic constraint $\delta_{\ind{A}}$ is ``whatever the law of $X$, it is only known that realisations will be inside $A$''.

\begin{remark}
The $\sigma$-algebra of $\mu_P$-measurable subsets is found to be the completion of the $\sigma$-algebra $\calX_B \defeq \{\emptyset, B, B^c, \boX\}$ with respect to any measure $m$ on $(\boX, \calX_B)$ verifying $m(B) > 0$ as well as $m(B^c) = 0$.
\end{remark}

\subsection{Upper bound}

If there exists $f \in \boL(\boX)$ such that $P = \delta_f$ then for any $B \in \calB(\boX)$,
\eqns{
p(B) \leq \int \Ssup{\ind{B} \cdot f'} P(\d f')  = \int \Ssup{\ind{B} \cdot f'} \delta_f(\d f') = \ssup{\ind{B} \cdot f}.
}
In this case, $\delta_f$ can be seen as the simplest non-trivial form of probabilistic constraint. The two previous examples are special cases with $f = \one$ and $f = \ind{A}$.

\begin{remark}
Probabilistic constraints of this form are equivalent to \emph{possibility distributions} \citep{Negoita1978, Dubois1988, Dubois1998, Dubois2000} and are also related to the notion of \emph{membership function} of a \emph{fuzzy set} \citep{Zadeh1965}. However, the approach considered in this work does not rely on the notion of fuzzy set and the form considered here is only used as a simple example of probabilistic constraint.
\end{remark}

\subsection{Combination of upper bounds}

If the probabilistic constraint $P$ is of the form $P = \sum_{i=1}^N a_i \delta_{f_i}$ then
\eqns{
p(B) \leq \int \Ssup{\ind{B} \cdot f'} P(\d f') = \sum_{i=1}^N a_i \ssup{\ind{B} \cdot f_i} = \sum_{i=1}^N a_i \sup_B f_i
}
for any $B \in \calB(\boX)$. This combination of upper bounds is not equivalent to a single bound in general and allows for modelling more accurate information. For instance, if $P = 0.5 \delta_{\ind{B}} + 0.5 \delta_{\ind{B'}}$ with $B \cap B' = \emptyset $, then one half of the probability mass of $p$ is in $B$ and the other half is in $B'$.

\subsection{Constraint based on a partition}

If there exists a measurable countable partition $\pi$ of $\boX$ and if the probability measure $P$ is of the form
\eqns{
P = \sum_{B \in \pi} q(B) \delta_{\ind{B}},
}
where $q$ is a probability measure on the sub-$\sigma$-algebra generated by $\pi$, then
\eqns{
\frall{\forall B \in \pi} p(B) = q(B),
}
i.e.\ the information available about $p$ is the one embedded into $q$.
\begin{proof}
From the definition of probabilistic constraint, we deduce that the inequality ${p(B) \leq q(B)}$ holds for any $B \in \pi$. Assume there exists $B \in \pi$ such that $p(B) < q(B)$. This assumption implies that
\eqns{
p(\boX) = \sum_{B \in \pi} p(B) < \sum_{B \in \pi} q(B) = q(\boX) = 1,
}
which is a contradiction since $p$ is a probability measure.
\end{proof}

Following the construction of an outer measure from a measure described in \eqref{eq:prop:probabilisticConstraint:inducedOuterMeas}, we define the outer measure $q^*$ induced by $q$ on $\boX$ as follows
\eqns{
\frall{\forall C \subseteq \boX} q^*(C) = \inf\{ q(B) \st B \in \sigma(\pi) \ET B \supseteq C \}.
}
Since $\pi$ is a partition of $\boX$, it is easy to check that the outer measure $\mu_P$ verifying
\eqns{
\mu_P(C) = \sum_{B \in \pi \sst B \cap C \neq \emptyset} q(B)
}
for all subsets $C$ of $\boX$ is equal to the outer measure $q^*$. This example shows that probabilistic constraints are versatile enough to model the same level of information as with a sub-$\sigma$-algebra generated by a partition.

\subsection{Probability measure}

Given that the singletons of $\boX$ are measurable, if the support of $P$ is in the subset $\boI_{\s}(\boX) \subset \boI(\boX)$ of indicator functions on singletons, then there exists a measure $q$ in $\boM_1(\boX)$ such that, for any $B \in \calB(\boX)$, it holds that
\eqns{
p(B) \leq \int_{\boI_{\s}(\boX)} \Ssup{\ind{B} \cdot f} P(\d f) = \int \ind{B}(x) q(\d x) = q(B), 
}
from which we conclude that $p = q$. The proof of the equality is very similar to the proof for constraints based on a partition.

\subsection{Plausibility}

If $P$ has the form
\eqns{
P = \sum_{A \in \calA} g(A) \delta_{\ind{A}}
}
for some set $\calA$ of measurable subsets and some function $g:\calA \to \bbR^+$, then
\eqns{
p(B) \leq \int_{\boI(\boX)} \Ssup{\ind{B} \cdot f} P(\d f) = \sum_{A \in \calA \sst A \cap B \neq \emptyset} g(A)
}
for any $B \in \calB(\boX)$, and the probabilistic constraint $P$ reduces to a \emph{plausibility} as defined in the context of Dempster-Shafer theory \citep{Dempster1967,Shafer1976}. A non-technical overview of the concepts of this theory is given by \cite[Sect.~4.4]{Williamson1989}. The two previous examples can be seen as special cases of plausibility where $g$ is a measure on $\calA = \pi$ or $\calA = \boI_{\s}(\boX)$.

The multiple cases given in this section show that probabilistic constraints can model information with various degrees of precision. In this work, we will be mainly interested in the most informative and the most uninformative cases.

\section{Operations on measure constraints}
\label{sec:operationsOnMeasureConstraints}

It is essential to be able to adapt a probabilistic constraint when the underlying probability measure is transformed, such as when considering a pushforward for the considered measure with respect to a given measurable function. In the following sections, $\boX_1$ and $\boX_2$ are assumed to be Polish spaces equipped with their Borel $\sigma$-algebra.

\subsection{Pushforward}

The operation of pushforwarding measures from one measurable space to another is central to measure theory and should therefore be translated in terms of measure constraints. In other words, if $M$ dominates the measure $m$ and $\xi$ is a measurable function, then what is the measure constraint dominating $\xi_*m$?

The objective is to find sufficient conditions for the existence of a measure constraint $M'$ verifying
\eqnl{eq:prop:pushforward}{
 M'(\chi (B, \cdot)) =  M(\chi(\xi^{-1}[B], \cdot ))
}
for all $B \in \calB(\boX_2)$, where $\chi$ is the characteristic function based on the supremum norm defined in \eqref{eq:characteristicSimple}. It first appears from the properties of outer measures described in Section~\ref{sec:outerMeasure} that $A \mapsto M(\chi(\xi^{-1}[A],\cdot))$ is an outer measure, so that $M'$ is possibly dominating $\xi_* m$. The existence of such a measure constraint is proved in the next proposition.

\begin{proposition}
\label{prop:pushforwardConstraint:equality}
Let $M$ be a measure constraint on $\boX_1$ and let $\xi$ be a measurable mapping from $\boX_1$ to $\boX_2$, then the measure constraint $M'$ on $\boX_2$ defined as the pushforward of $M$ by the mapping $T_{\xi}$ from $\boL^{\infty}(\boX_1)$ to $\boL^{\infty}(\boX_2)$ characterised by
\eqns{
\frall{\forall f \in \boL^{\infty}(\boX_1)} T_{\xi}(f) : y \mapsto \sup_{\xi^{-1}[\{y\}]} f
}
verifies \eqref{eq:prop:pushforward}.
\end{proposition}

\begin{proof}
The first step is to rewrite the uniform norm over $\boX_1$ in a suitable way as
\eqnsa{
\| \ind{\xi^{-1}[B]} \cdot f \| & = \sup_{y \in \boX_2} \Big( \sup_{x \in \xi^{-1}[\{y\}]} \big( \ind{\xi^{-1}[B]}(x) f(x) \big) \Big) \\
& = \sup_{y \in \boX_2} \Big( \ind{B}(y) \sup_{\xi^{-1}[\{y\}]} f \Big)
}
for any $f \in \boL^{\infty}(\boX_1)$ and any $B \in \calB(\boX_2)$, where the second line is explained by the fact that the function $\ind{\xi^{-1}[B]}$ is constant over $\xi^{-1}[\{y\}]$, and is equal to $\ind{B}(y)$ everywhere on this subset. By \cite[Corollary~2.13]{Crauel2003}, it holds that $T_{\xi}(f) \in \boL^{\infty}(\boX_2)$ since $\{ \xi^{-1}[\{y\}] \times \{y\} : y \in B \}$ is a measurable subset of $\boX_1 \times \boX_2$ for any $B \in \calB(\boX_2)$. The fact that $T_{\xi}$ is measurable follows from the assumption that $\sup : \boL^{\infty}(\boX) \to \bbR$ is measurable. It then holds that
\eqns{
\int \| \ind{B} \cdot  f' \| M'(\d f') = \int \| \ind{\xi^{-1}[B]} \cdot f \| M(\d f),
}
where $M'$ is defined as the pushforward $(T_{\xi})_*M$.
\end{proof}

\begin{remark}
Proposition~\ref{prop:pushforwardConstraint:equality} can be straightforwardly generalised to characteristic functions of the form $\chi : (B,f) \mapsto \Ssup{(\chi' \circ \ind{B})\cdot f}$ for some suitable function $\chi' : \bbR \to \bbR$.
\end{remark}

In the following example, the operation of marginalisation is translated to probabilistic constraints.

\begin{example}
Let $\boX$ be the Cartesian product of $\boX_1$ and $\boX_2$ and let $p$ be a given probability measure on $\boX$. The objective is to project a given probabilistic constraint $P \in \boC_1(\boX)$ for $p$ into a probabilistic constraint $P'$ in $\boC_1(\boX_2)$ for the corresponding marginal $p' \in \boM_1(\boX_2)$ characterised by $p'(B) = p(\boX_1 \times B)$. Marginalisation can be performed on $P'$ by using Proposition~\ref{prop:pushforwardConstraint:equality} with the canonical projection map ${\xi : \boX \to \boX_2}$, from which we find that
\eqns{
\frall{\forall B \in \calB(\boX_2)} p(B) \leq \int \Ssup{\ind{B} \cdot f'} (T_{\xi})_*P(\d f'),
}
where the mapping $T_{\xi} : \boL^{\infty}(\boX) \to \boL^{\infty}(\boX_2)$ is found to be
\eqns{
T_{\xi}(f) : y \mapsto \sup_{\xi^{-1}[\{y\}]} f = \Ssup{f(\cdot,y)}.
}
\end{example}

Another case, studied in the next example, is the pushforward of a probabilistic constraint defined as being equivalent to a probability measure on a sub-$\sigma$-algebra.

\begin{example}
\label{ex:pushforwardSubSigAlg}
Considering again the case detailed in Section~\ref{sec:exampleProbaConstraint} where $\pi$ is a measurable countable partition of $\boX_1$ and where the probabilistic constraint $P$ is of the form
\eqns{
P = \sum_{B \in \pi} p'(B) \delta_{\ind{B}},
}
where $p'$ is a probability measure on the sub-$\sigma$-algebra generated by $\pi$, then for any measurable mapping $\xi : \boX_1 \to \boX_2$ such that $\sigma(\xi) \subseteq \sigma(\pi)$, it holds that
\eqns{
\frall{\forall B \in \pi}T_{\xi}(\ind{B}) : y \mapsto \sup_{\xi^{-1}[\{y\}]} \ind{B} = \ind{\xi(B)}(y),
}
that is, since $\xi[B]$ is a singleton,
\eqns{
T_{\xi}(\ind{B}) = \ind{\xi(B)} \in \boI_{\s}(\boX_2),
}
so that $(T_{\xi})_*P$ has its support in $\boI_{\s}(\boX_2)$ and is therefore equivalent to a probability measure on $\boX_2$. 
\end{example}

The concepts of pushforward and marginalisation for probabilistic constraints are important and will be useful in practice. We now consider the converse operation, referred to as pullback and which will also contribute to the formulation of the main results in this work.

\subsection{Pullback}

The objective in this section is to understand in which situations it is possible to reverse the operation of pushforwarding measures and measure constraints. We first define what is understood as a pullback.

\begin{definition}
Let $m$ be a measure on the space $\boX_2$ and let $\xi$ be a measurable mapping from $\boX_1$ to $\boX_2$, then a \emph{pullback} measure, denoted $\xi^*m$, is a measure in $\boM_1(\boX_1)$ satisfying
\eqns{
\xi_*(\xi^*m) = m.
}
\end{definition}

Note that there are possibly many pullback measures for a given measure so that the concept of measure constraint will be useful in order to represent this multiplicity, as in the following proposition. As with the pushforward, we want to obtain the existence of a measure constraint $M'$ with characteristic function $\chi$ for the pullback of a measure $m \in \boM_1(\boX_2)$ such that
\eqnl{eq:prop:pullback}{
M'(\chi(\xi^{-1}[B],\cdot)) = M(\chi(B,\cdot))
}
holds for any $B \in \calB(\boX_2)$, where $M$ is a measure constraint dominating $m$. It appears from the properties of outer measures described in Section~\ref{sec:outerMeasure} that $A \mapsto M(\chi(\xi[A],\cdot))$ is an outer measure, so that $M'$ is possibly a measure constraint dominating the pullback measures of $m$.

\begin{proposition}
\label{prop:pullbackConstraint:equality}
Let $M$ be a measure constraint on $\boX_2$ and let $\xi$ be a measurable mapping from $\boX_1$ to $\boX_2$, then the measure constraint $M'$ on $\boX_1$ defined as the pushforward of $M$ by the mapping $T'_{\xi} : \boL^{\infty}(\boX_2) \to \boL^{\infty}(\boX_1)$ defined as
\eqns{
T'_{\xi} : f \mapsto f \circ \xi
}
verifies \eqref{eq:prop:pullback}.
\end{proposition}

\begin{proof}
We rewrite the uniform norm of $\ind{B} \cdot f$ as
\eqns{
\Ssup{\ind{B} \cdot f} = \Ssup{\ind{\xi^{-1}[B]} \cdot (f \circ \xi)}
}
for any $f \in \boL^{\infty}(\boX_2)$ and any $B \in \calB(\boX_2)$. Since the composition of measurable mappings is measurable, we verify that the codomain of $T'_{\xi}$ is $\boL^{\infty}(\boX_1)$. Since the mapping $T'_{\xi}$ is assumed to be measurable, we write
\eqns{
\int \Ssup{\ind{B} \cdot f} M(\d f) = \int \Ssup{\ind{\xi^{-1}[B]} \cdot f'} (T'_{\xi})_* M(\d f'),
}
which terminates the proof of the proposition.
\end{proof}

\begin{remark}
Other transformations than $T'_{\xi}$ could lead to a measure on $\boL^{\infty}(\boX_1)$ that verifies \eqref{eq:prop:pullback}. However, these measures would not dominate all the pullback measures of $m$.
\end{remark}

\begin{example}
If $m$ is a known probability measure, then the corresponding canonical probabilistic constraint $M$ would have its support in the set $\boI_{\s}(\boX_2)$ of indicator functions of singletons in $\boX_2$. The induced probabilistic constraint $M'$ on $\boX_1$ would then be of the form
\eqns{
M' = \sum_{B \in \pi} m'(B) \delta_{\ind{B}},
}
where $\pi$ is the partition generated by $\xi$ on $\boX_2$ and $m'$ is the probability measure induced by $m$ on the sub-$\sigma$-algebra generated by $\xi$. This example closes a loop started in Example~\ref{ex:pushforwardSubSigAlg} where a probabilistic constraint on a sub-$\sigma$-algebra generated by a partition was shown to induce a probability measure on $\boX_2$.
\end{example}

Now equipped with the concepts of pushforward and pullback for measure constraints and measures, we can address the question of data assimilation between different sources of information.

\section{Data assimilation with probabilistic constraints}
\label{sec:fusion}

The objective in this section is to introduce a data-assimilation operation for independent sources of information about the same physical system.

\subsection{State-space fusion}
\label{ssec:stateSpaceFusion}

We first address the case where the state satisfies some desirable but stringent conditions before showing that a much larger class of state spaces can be considered under a weak assumption. We assume that there exists a \emph{representative} set $\bscalX$ in which the state of the system of interest can be uniquely characterised and we denote $\xi$ the projection map between $\bscalX$ and $\boX$. Both $\bscalX$ and $\boX$ are assumed to be Polish spaces equipped with their Borel $\sigma$-algebras.

Two probabilistic constraints $P,P' \in \boC_1(\bscalX)$ are considered and assumed to be \emph{compatible}, that is, at most one of them corresponds to an unknown probability measure. These probabilistic constraints can be understood as either prior knowledge and observation or both as observations, which justifies the use of the generic term \emph{data} to cover the two types of information. Although one or both of these probabilistic constraint do not represent a true law, it is useful to consider two probability measures respectively dominated by $P$ and $P'$ in order to unveil the mechanism behind data assimilation.

Let $X \in \boL^0(\Omega,\boX)$ and $X' \in \boL^0(\Omega',\boX)$ be two random variables in $\boX$ constrained by $P$ and $P'$ respectively and consider the predicate
\eqns{
\tag{\#}
\mbox{``$X$ and $X'$ represent the same physical system''.}
}
The random variables $X$ and $X'$ are defined on different probability spaces as they are assumed to originate from different representations of randomness. As a consequence, the predicate $(\#)$ cannot be expressed as $X = X'$. A random variable $X \times X'$ can be defined on the probability space $(\Omega \times \Omega', \Sigma \otimes \Sigma', \bbP \rtimes \bbP')$ of joint outcomes/events and the law of $X \times X'$ is the probability measure $p \rtimes p'$ where $p$ and $p'$ are the respective laws of $X$ and $X'$. 
The predicate $(\#)$ cannot be expressed as an event in $\boX \times \boX$ in general, yet, it can be expressed as the event $\Delta = \{(x,x) \st x \in \bscalX \}$ in the product $\sigma$-algebra $\calB(\bscalX)\otimes \calB(\bscalX)$ because of the representativity of the set $\bscalX$.

The data-assimilation operation for independent sources of information can now be formalised as an operation between probabilistic constraints as shown in the following theorem.

\begin{theorem}
\label{thm:atomicFusion}
Let $\bscalX$ be a representative set and let $P$ and $P'$ be compatible probabilistic constraints on $\bscalX$ assumed to represent the same physical system, then the probabilistic constraint resulting from data assimilation is characterised by the following convolution of measures on $(\boL^{\infty}(\bscalX),\cdot)$:
\eqnl{eq:thm:atomicFusion:convolution}{
P \star P' = \dfrac{(P * P')^{\scl}}{\| P * P'\|},
}
whenever $\| P * P'\| \neq 0$ holds.
\end{theorem}

The use of the unary operation $\cdot^{\scl}$ in the numerator of \eqref{eq:thm:atomicFusion:convolution} has for only objective to make $P \star P'$ a canonical probabilistic constraint. An equivalent choice would be to define $P \star P'$ as the normalised version of $P * P'$. However, \eqref{eq:thm:atomicFusion:convolution} is preferred since it ensures that $P \star P'$ is also a probability measure on $\boL^{\infty}(\bscalX)$. A similar rule has been proposed without proof in \cite{Mahler2005} in the context of fuzzy Dempster-Shafer theory \citep{Yen1990}.

\begin{proof}
Let $X \in \boL^0(\Omega,\bscalX)$ and $X' \in \boL^0(\Omega',\bscalX)$ be two random variables in $\bscalX$ which laws $p = X_*\bbP$ and $p' = X'_*\bbP'$ are respectively dominated by $P$ and $P'$. The main idea of the proof is to build the joint law $p \rtimes p'$ in $\bscalX\times\bscalX$ and to study the diagonal $\Delta = \{(x,x) \st x \in \bscalX\}$. The data assimilation mechanism for probabilistic constraints can then be deduced from the one for $p$ and $p'$. Let $\Delta(B)$ be the intersection $\Delta \cap B \times B$ for any $B \in \calB(\bscalX)$ and consider the law $p \rtimes p'$ verifying
\eqns{
p \rtimes p'(B \times B') \leq \int \Ssup{\ind{B \times B'} \cdot (f \rtimes f')} P(\d f) P'(\d f'),
}
for any $B \times B' \in \calB(\bscalX\times\bscalX)$ so that
\eqnsa{
p \rtimes p'(\Delta(B)) & \leq \int \Ssup{\ind{\Delta(B)} \cdot (f \rtimes f')} P(\d f) P'(\d f') \\
& \leq \int \| \ind{B} \cdot (f \cdot f')^{\scl} \| \| f \cdot f' \| P(\d f) P'(\d f') \\
& \leq \int \| \ind{B} \cdot \hat{f} \| M(\d \hat{f})
}
holds for any $B \in \calB(\bscalX)$, with
\eqns{
M(F) \defeq \int \ind{F}((f \cdot f')^{\scl}) \| f \cdot f' \| P(\d f) P'(\d f') = (P * P')^{\scl}(F)
}
for any measurable subset $F$ of $\boL^{\infty}(\bscalX)$. If both $p$ and $p'$ were the actual probability laws depicting the two sources of information, it would be necessary to find a lower bound for $p \rtimes p'(\Delta)$ in order to determine the constraint for the hypothetical posterior law given $\Delta$ (note that such a lower bound would be equal to zero in most of the situations). On the contrary, it has been assumed that at least one of the probabilistic constraints is a primary representation of uncertainty so that it is the measure constraint $M$ that has to be normalised. We conclude that the probabilistic constraint $P \star P'$ on $\bscalX$ verifies
\eqns{
P \star P'(F) = \dfrac{1}{C} (P * P')^{\scl}(F)
}
for any $F \in \calL^{\infty}(\bscalX)$, 
where the normalising constant is found to be
\eqns{
C = (P * P')^{\scl}(\boL^{\infty}(\boX)) 
= \int \| f \cdot f' \| P(\d f) P'(\d f') = \| P * P' \|.
}
The convolution of $P$ and $P'$ is well defined since $(\boL^{\infty}(\bscalX),\cdot)$ is a topological semigroup, as shown in \cite[Sect.~245D]{Fremlin2000}.
\end{proof}


\begin{remark}
\label{rem:stateSpaceFusionBayes}
Using the notations of Theorem~\ref{thm:atomicFusion} and defining the likelihood function $\ell(\Delta \given \cdot)$ on $\boL^{\infty}(\bscalX) \times \boL^{\infty}(\bscalX)$ as
\eqns{
\ell(\Delta \given \cdot) : (f,f') \mapsto \ssup{f \cdot f'},
}
the probability measure $P \star P'$ can be expressed for any $F \in \calL^{\infty}(\bscalX)$ as
\eqns{
P \star P'(F) = \dfrac{P \rtimes P'(\ell(\Delta \given \cdot) \Phi(\cdot,F))}{P \rtimes P'(\ell(\Delta \given \cdot))},
}
where $\Phi$ is a Markov kernel from $\boL^{\infty}(\bscalX) \times \boL^{\infty}(\bscalX)$ to $\boL^{\infty}(\bscalX)$ defined as
\eqns{
\Phi : ((f,f'),F) \mapsto  \delta_{(f \cdot f')^{\scl}}(F).
}
This shows that the canonical version of $P \star P'$ is a proper Bayes' posterior probability measure on $\boL^{\infty}(\bscalX)$. This is a fundamental result since it shows that general data assimilation can be performed in a fully measure-theoretic Bayesian paradigm.
\end{remark}

One immediate objection to Theorem~\ref{thm:atomicFusion} is that the assumption that the set $\bscalX$ is representative is highly restrictive. Very often there is no interest in handling sophisticated state spaces, especially when the received information does not allow for such an in-depth representation. However, several systems could have the same state in $\boX$ and $(\#)$ cannot be expressed as an event in $\calB(\boX) \otimes \calB(\boX)$ so that a different approach has to be considered. In the next corollary, we show that the result of Theorem~\ref{thm:atomicFusion} can be extended to a simpler state space as long as it verifies a natural assumption. 

\begin{corollary}
\label{cor:fusion}
Assuming that there exists a representative set in which the physical system can be uniquely characterised, the result of Theorem~\ref{thm:atomicFusion} holds for random variables in the state space $(\boX,\calB(\boX))$.
\end{corollary}

\begin{proof}
Let $\xi$ be a given projection map from $\bscalX$ to $\boX$. The operations of pushforward and pullback can be respectively expressed via the following mappings on the space of measurable functions: 
\eqns{
\frall{\forall f \in \boL^{\infty}(\bscalX)} T_{\xi}(f): x \mapsto \sup_{\xi^{-1}[\{x\}]} f,
}
and
\eqns{
\frall{\forall f \in \boL^{\infty}(\boX)} T'_{\xi}(f) = f \circ \xi.
}
We compute the result of the pusforward of combined pullback functions as follows:
\eqns{
\frall{\forall f,f' \in \boL^{\infty}(\boX)} T_{\xi} \big(T'_{\xi}(f) \cdot T'_{\xi}(f') \big) = T_{\xi}( (f \cdot f') \circ \xi ),
}
so that
\eqns{
T_{\xi}\big( T'_{\xi}(f) \cdot T'_{\xi}(f') \big) : x \mapsto \sup_{\xi^{-1}[\{x\}]} T_{\xi}( (f \cdot f') \circ \xi ) = f \cdot f'.
}
This result indicates that the pushforward and the pullback via $\xi$ do not need to be considered when combining probabilistic constraints on the non-representative state space $\boX$ and that, as a result, data assimilation with probabilistic constraints on $\boX$ and on $\bscalX$ can be performed in the same way. 
\end{proof}

The following examples make use of the notations of Theorem~\ref{thm:atomicFusion} and detail two simple cases for which only one of the elements in the probabilistic constraint is maintained through the data assimilation operation.

\begin{example}
If $P$ has its support in $\boI_{\s}(\boX)$ and is therefore equivalent to $p \in \boM_1(\boX)$ and if $P'$ is of the form $P' = \delta_{f'}$, then the probabilistic constraint $P \star P' \in \boC_1(\boX)$ satisfies
\eqns{
P \star P'(F) \propto \int \ind{F}((f'(x)\ind{\{x\}})^{\scl}) f'(x) p(\d x).
}
Since $f'(x)\ind{\{x\}}$ is either the null function or is supported by a singleton, $P\star P'$ has its support in $\boI_{\s}(\boX)$ and is equivalent to the measure $\hat{p} \in \boM_1(\boX)$ defined as
\eqns{
\hat{p}(\d x) \inteq \dfrac{1}{p(f')} f'(x) p(\d x),
}
that is, the law of the combined random variable is known to be $\hat{p}$. The same type of result holds if the roles of $P$ and $P'$ are interchanged. If $f'$ is actually a likelihood of the form $\ell_z$, then $\hat{p}$ can be expressed as
\eqns{
\hat{p}(\d x) \inteq \dfrac{\ell_z(x) p(\d x)}{\int \ell_z( x') p(\d x')},
}
which is the usual Bayes' posterior of $p$ given $z$, where $z$ is interpreted as an observation and $p$ is the prior probability measure. The fact that $\ell_z$ is an element of $\boL(\boX)$ rather than a probability density does not affect the result because of the normalisation factor $\int \ell_z( x') p(\d x')$. As mentioned in Example~\ref{ex:HMM}, one convenient choice for $\ell_z$ is the Gaussian-shaped upper bound defined in the one-dimensional observation case as
\eqns{
\ell_z(x) \defeq \exp\bigg( -\dfrac{(Hx-z)^2}{2\sigma^2} \bigg),
}
where $\sigma$ corresponds to the standard deviation of the corresponding normal distribution and where $H$ is the observation matrix. This choice allows for keeping the Kalman filter recursion while better representing the information received via finite-resolution sensors. Solutions to the multi-target tracking problem in the linear Gaussian case can be formulated in terms of Kalman filters in interaction \cite{DelMoral2013,DelMoral2015} and depend on the value of the denominator in Bayes' theorem which would usually be of the form
\eqns{
\dfrac{1}{\sqrt{2\pi}\varsigma} \exp\bigg( -\dfrac{(Hm-z)^2}{2\varsigma^2} \bigg),
}
with $\varsigma^2 = HPH^{\T} + \sigma^2$, where $m$ and $P$ are the mean of variance of the prior distribution and where $\cdot^{\T}$ is the matrix transposition. With the proposed solution, we find that
\eqns{
\int \ell_z(x) p(\d x) = \dfrac{\sigma}{\varsigma} \exp\bigg( -\dfrac{(Hm-z)^2}{2\varsigma^2} \bigg),
}
which is dimensionless, takes value in the interval $[0,1]$ and can be interpreted as the probability for the observation $z$ to belong to the target with distribution~$p$.
\end{example}

\begin{example}
If $P$ and $P'$ have the form $P = \delta_f$ and $P' = \delta_{f'}$, then the probabilistic constraint $P\star P' \in \boC_1(\boX)$ is found to be
\eqns{
P \star P' = \delta_{(f \cdot f')^{\scl}}
}
that is, $P \star P'$ is a canonical probabilistic constraint and $(f \cdot f')^{\scl} \in \boL(\boX)$ can be seen as the resulting upper bound.
\end{example}

The binary operation $\star$ on $\boC_1(\boX)$ introduced in Theorem~\ref{thm:atomicFusion} has some additional properties that are detailed in the following theorem.

\begin{theorem}
\label{thm:indFusionMonoid}
The space $(\boC_1(\boX), \star)$ is a commutative semigroup with an identity element.
\end{theorem}

\begin{proof}
Theorem~\ref{thm:atomicFusion} together with Corollary~\ref{cor:fusion} show that $\star$ is a proper binary operation on $\boC_1(\boX)$ in the sense that it is found to be a relation from $\boC_1(\boX) \times \boC_1(\boX)$ to $\boC_1(\boX)$. In order to prove the result of the theorem, we need to show that $\star$ is also associative, has an identity element in $\boC_1(\boX)$ and is commutative: let $P,P'$ and $P''$ be probabilistic constraints in $\boC_1(\boX)$, then
\begin{enumerate}[label=\itshape\alph*\upshape)]
\item $\star$ is associative: Since the convolution of measures $*$ is known to be associative, we need to show that $((P * P')^{\scl} * P'')^{\scl} = (P * (P' * P'')^{\scl})^{\scl}$ holds. This can be done by verifying that
\eqns{
\| f \cdot f' \| \| \hat{f}  \cdot f'' \| = \| f \cdot f' \cdot f'' \|
}
and
\eqns{
\dfrac{ \hat{f} \cdot f''}{\| \hat{f}  \cdot f'' \|} = \dfrac{f \cdot f' \cdot f''}{\| f \cdot f' \cdot f'' \|}
}
hold with $\hat{f} = \frac{f \cdot f'}{\|f \cdot f'\|}$ for any $f,f',f'' \in \boL^{\infty}(\boX)$.
\item $\star$ has an identity element: consider that $P' = \delta_{\one}$, then for any measurable subset $F \in \calL^{\infty}(\boX)$
\eqns{
P \star P'(F) \propto \int \ind{F}((f \cdot \one)^{\scl}) \| f \| P(\d f) = P(F),
}
with the coefficient of proportionality
\eqns{
\int \|f \cdot \one\| P(\d f) = 1,
}
so that $P \star P' = P$ and $P'$ is the identity element of $(\boC_1(\boX), \star)$.
%
\end{enumerate}
This terminates the proof.
\end{proof}

The absence of inverse element for the operator $\star$ on $\boC_1(\boX)$ is due to the fact that it is not always possible to ``forget'' what has been learnt. We now study the relation between the proposed operation and Dempster's rule of combination in the next example.

\begin{example}
Let $P$ and $P'$ be two canonical probabilistic constraints on $\boX$ of the form
\eqns{
P = \sum_{A \in \calA} m(A) \delta_{\ind{A}} \AND P' = \sum_{A' \in \calA'} m'(A') \delta_{\ind{A'}}
}
for some sets $\calA$ and $\calA'$ of measurable subsets and some functions $m:\calA \to [0,1]$ and $m':\calA' \to [0,1]$. The posterior probabilistic constraint $P \star P'$ on $\boX$ verifies
\eqns{
P \star P'(F) \propto \sum_{A \cap A' \neq \emptyset} \ind{F}(\ind{A \cap A'}) m(A) m'(A')
}
for any $F \in \calL^{\infty}(\boX)$, which can be re-expressed as
\eqns{
P \star P' = \dfrac{1}{\| P * P'\|} \sum_{A \cap A' \neq \emptyset} m(A) m'(A') \delta_{\ind{A \cap A'}},
}
with
\eqns{
\| P * P'\| = \sum_{A \cap A' \neq \emptyset} m(A) m'(A') = 1- \sum_{A \cap A' = \emptyset} m(A)m(A'),
}
so that $P \star P'$ reduces to the result of Dempster's rule of combination between $m$ and $m'$ \citep{Dempster1968,Shafer1986}. 
A discussion about the connections between Dempster's rule of combination and Bayes' theorem applied to second-order probabilities can be found in \cite{Baron1987}.
\end{example}

\subsection{General fusion}

In Section~\ref{ssec:stateSpaceFusion}, it has been assumed that the set $\boX$ on which the probabilistic constraints are defined can be understood as a state space, with the consequence that the system of interest should be represented at a single point of $\boX$. The way elements of $\boL^{\infty}(\boX)$ are fused is highly dependent on this assumption. Yet, a different viewpoint must be considered when introducing probabilistic constraints on more general sets such as $\boM_1(\boX)$ and $\boC_1(\boX)$.

\begin{remark}
In order to have a suitable topology defined on $\boL^0(\boC_1(\boX))$, a reference measure must be introduced. As there is no natural reference measure on $\boC_1(\boX)$, a case-by-case reference measure must be defined, e.g., as an integer-valued measure that singles out a countable number of elements of $\boC_1(\boX)$, or as a parametric family of elements of $\boC_1(\boX)$ which parameter lies in an Euclidean space.
\end{remark}

Let $\boY$ be a topological space equipped with its Borel $\sigma$-algebra and assume there is a given way of combining elements of $\boY$ which is characterised by a potential function $\ell : \boY \times \boY \to [0,1]$ and a surjective mapping $\theta : S \to \boY$, with $S \subseteq \boY \times \boY$, which defines a stochastic kernel $\Phi$ from $\boY \times \boY$ to $\boY$ as
\eqns{
\Phi((y,y'), \cdot) \defeq
\begin{cases*}
\delta_{\theta(y,y')} & if $(y,y') \in S$ \\
\zero & otherwise.
\end{cases*}
}
We also assume that $\ell$ verifies $\ell(y,y') = 0$ for any $(y,y') \in S^c$. Note that the kernel $\Phi$ becomes a Markov kernel when restricted to $S$. When it exists, the probability measure $\hat\bsp$ in $\boM_1(\boY)$ defined as
\eqnl{eq:generalFusionProba}{
\hat\bsp(B) = \dfrac{\bsp\rtimes \bsp'(\ell \Phi(\cdot ,B))}{\bsp\rtimes \bsp'(\ell)},
}
for all $B \in \calB(\boY)$ can be introduced. As in Remark~\ref{rem:stateSpaceFusionBayes}, $\hat\bsp$ can be seen as the projection of the Bayes' posterior corresponding to the prior $\bsp \rtimes \bsp'$ and the likelihood $\ell$.

\begin{example}
If $\boY$ is equal to $\boL^{\infty}(\bscalX)$, with $\bscalX$ the representative state space defined in Section~\ref{ssec:stateSpaceFusion}, then the laws $\bsp$ and $\bsp'$ are probabilistic constraints on $\bscalX$. In this case, the natural choice for $\ell$ and $\theta$ is $\ell(f,f') = \Ssup{f \cdot f'}$ and $\theta(f,f') = (f \cdot f')^{\scl}$ for any $f, f' \in \boL^{\infty}(\bscalX)$. The probability measure $\hat\bsp$ then takes the form
\eqns{
\hat\bsp(F) \propto \int \ind{F}((f \cdot f')^{\scl}) \Ssup{f\cdot f'} \bsp(\d f) \bsp'(\d f'),
}
for all $F \in \calL^{\infty}(\boY)$, that is $\hat\bsp = \bsp \star \bsp'$.
\end{example}

We assume that there is a reference measure on $\boY$ which enables probability measures to be defined on $\boL^{\infty}(\boY)$. The objective is to extend the fusion of probabilistic constraints to the case where $\boY$ is an arbitrary set.

\begin{theorem}
\label{thm:fuse2ndOrder}
Let $\bsP$ and $\bsP'$ be probabilistic constraints on $\boY$ assumed to represent the same physical system, then the probabilistic constraint resulting from data assimilation is characterised by the following convolution of measures on $(\boL^{\infty}(\bscalX),\fpstar)$:
\eqnl{thm:fuse2ndOrder:eq:pstar}{
\bsP \star \bsP' = \dfrac{(\bsP * \bsP')^{\scl}}{\| \bsP * \bsP' \|},
}
where the binary operation $\fpstar$ on $\boL^{\infty}(\boY)$ characterised by
\eqnl{thm:fuse2ndOrder:eq:star}{
f \fpstar f' : \hat{y} \mapsto \sup_{(y,y') \in \theta^{-1}[\{\hat{y}\}]} \ell(y,y') f(y) f'(y')
}
is assumed to be associative.
\end{theorem}

The level of generality of Theorem~\ref{thm:fuse2ndOrder} is necessary for tackling estimation problems for systems that display hierarchical levels of uncertainty, such as in the modelling of stochastic populations \cite{Houssineau2016}.

\begin{proof}
Let $\bsp$ and $\bsp'$ be probability measures on $\boY$ dominated by $\bsP$ and $\bsP'$ respectively. Equation \eqref{eq:generalFusionProba} shows that the measure $\bsm$ to be dominated is characterised by
\eqns{
\frall{\forall B \in \calB(\boY)} \bsm(B) = \bsp\rtimes \bsp'(\ell\, \Phi(\cdot ,B)),
}
and we find that
\eqns{
\bsm(B) \leq \int \sup_{(y,y') \sst \theta(y,y') \in B}\big( \ell(y,y') f(y) f'(y') \big) \bsP(\d f) \bsP'(\d f')
}
for all $B \in \calB(\boY)$. In order to express the right hand side of the previous inequality as a probabilistic constraint, the argument of the supremum has to be normalised, and we find that
\eqnsa{
\bsP \star \bsP'(F) & = \dfrac{1}{C} \int \ind{F}((f \fpstar f')^{\scl}) \| f \fpstar f' \| \bsP(\d f) \bsP'(\d f') \\
& \propto (\bsP * \bsP')^{\scl}(F)
}
for all $F \in \calL^{\infty}(\boY)$, where the operation~$\fpstar$ is defined as in the statement of the theorem. The constant $C$ is found to be
\eqns{
C = \int \| f \fpstar f' \| \bsP(\d f) \bsP'(\d f') = \| \bsP * \bsP' \|,
}
thus proving the result of the theorem.
\end{proof}

For the posterior probabilistic constraint of Theorem~\ref{thm:fuse2ndOrder} to be well defined, the mappings $\theta$ and $\ell$ have to have sufficient properties for $(\boL^{\infty}(\boY), \fpstar)$ to be a semigroup. Informally, the mapping $\theta$ can be loosely seen as a binary operation which has to be associative for $\fpstar$ to be associative. More rigorously, we can define an extension $\bar\boY \defeq \boY \cup \{\vphi\}$ of the set $\boY$ by an isolated point $\vphi$ and extend $\theta$ and $\ell$ as follows: $\theta(y,y') = \vphi$ for any $(y,y') \notin S$ and $\ell(y,y') = 0$ if $(y,y') \notin \boY \times \boY$. With these notations, we can formulate the following proposition.

\begin{proposition}
The binary operation $\fpstar$ is associative if $\theta$ is associative and if
\eqns{
\ell(y,y') \ell(\theta(y,y'),y'') = \ell(y,\theta(y',y'')) \ell(y',y'')
}
holds for any $y,y',y'' \in \boY$.
\end{proposition}

\begin{proof}
We have to show that $(f \fpstar f') \fpstar f'' = f \fpstar (f' \fpstar f'')$ for any $f,f',f'' \in \boL^{\infty}(\boY)$ under the conditions given in the proposition. It holds that for any $\tilde{y} \in \boY$
\eqnsml{
((f \fpstar f') \fpstar f'')(\tilde{y}) = \\
\sup_{\substack{(\hat{y},y'') \in \theta^{-1}[\{\tilde{y}\}] \\ (y,y') \in \theta^{-1}[\{\hat{y}\}] }} \ell(y,y') \ell(\hat{y},y'') f(y) f'(y') f''(y''),
}
which can be expressed as
\eqnsml{
((f \fpstar f') \fpstar f'')(\tilde{y}) = \\
\sup_{(y,y',y'') \sst \theta(\theta(y,y'),y'') = \tilde{y}} \ell(y,y') \ell(\theta(y,y'),y'') f(y) f'(y') f''(y'').
}
Obtaining the equivalent expression for $f \fpstar (f' \fpstar f'')$ confirms that $\fpstar$ is associative if
\eqnsa{
\theta(\theta(y,y'),y'') & = \theta(y,\theta(y',y'')) \\
\ell(y,y') \ell(\theta(y,y'),y'') & = \ell(y,\theta(y',y'')) \ell(y',y'')
}
holds for any $y,y',y'' \in \boY$, the first line being equivalent to the associativity of $\theta$ when seen as a binary operation.
\end{proof}

Notice that when the mapping $\theta$ is bijective, the function $f \fpstar f'$ is also characterised by the relation
\eqnl{eq:bijTheta}{
f \fpstar f'(\theta(y,y')) = \ell(y,y')f(y)f'(y')
}
for all $(y,y') \in S$. Indeed, the supremum in \eqref{thm:fuse2ndOrder:eq:star} is taken over $\theta^{-1}[\{\hat{y}\}]$ which is a singleton when $\theta$ is bijective. The following example makes the connection between state-space fusion and the more general fusion operation introduced in Theorem~\ref{thm:fuse2ndOrder}.

\begin{example}
\label{ex:generalFusion:ex1} If $\boY$ is equal to the state space $\boX$ defined in Section~\ref{ssec:stateSpaceFusion}, then the natural way to set $\ell$ and $\theta$ is to take $\ell(x,x') = \ind{\{x\}}(x')$ for all $x,x' \in \boX$ and to define $\theta$ on the diagonal of $\boX \times \boX$ only, by $\theta(x,x) = x$ for all $x \in \boX$. Since $\theta$ is bijective in this case, we can use \eqref{eq:bijTheta} to find that
\eqns{
f \fpstar f' : x \mapsto \ell(x,x) f(x) f'(x) = f(x)f'(x).
}
In other words, it holds that $f \fpstar f' = f \cdot f'$ and the result of Section~\ref{ssec:stateSpaceFusion} about fusion for state spaces is recovered. This confirms that the fusion operation introduced in this section is not different from the one of Section~\ref{ssec:stateSpaceFusion}, it is instead a more general formulation of the same operation.
\end{example}

\begin{example} If $\boY = \boC_1(\boX)$ where $\boX$ is the state space defined in Section~\ref{ssec:stateSpaceFusion}, then the natural way to set $\ell$ and $\theta$ is $\ell(P,P') = \| P * P' \|$ and $\theta(P,P') = P \star P'$ for any probabilistic constraints $P$ and $P'$ in $\boC_1(\boX)$. In this case, we find that
\eqns{
f \fpstar f' : \hat{P} \mapsto \sup_{(P,P') \sst P \star P' = \hat{P}} \|P * P'\| f(P) f'(P').
}
The binary operation $\fpstar$ is associative since it indeed holds that
\eqnsa{
(P \star P') \star P'' & = P \star (P' \star P'') \\
\| P * P'\| \| ( P \star P') * P'' \| & = \| P * (P' \star P'') \| \| P' * P''\|
}
for any $P,P',P'' \in \boC_1(\boX)$. The set $\boC_1(\boX)$ is one of the useful examples of sets on which the fusion operation takes a general form since two probabilistic constraints do not have to be equal to represent the same individual. Also, the mapping $\theta$ is surjective but not bijective in general for probabilistic constraints, as opposed to the case of Example~\ref{ex:generalFusion:ex1}. This level of generality was required in \cite[Chapt.\ 2]{Houssineau2015} when deriving a principled solution to the problem of data association for multi-object dynamical systems since these systems are best represented by hierarchical level of uncertainty, as already identified in \cite{Pace2013,DelMoral2015}.
\end{example}

\section*{Conclusion}

Supremum-based outer measures have demonstrated their ability to represent uncertainty in a flexible way and to comply with intuitively appealing operations such as pullback and data assimilation. Future work includes the generalisation of the proposed type of outer measure to product spaces, where correlations between the different components of the product space can take a more involved form.

\bibliography{Thesis}

\end{document}

%% file: localDefs.tex
\newtheorem{theorem}{Theorem}
\newtheorem{corollary}{Corollary}
\newtheorem{proposition}{Proposition}

\theoremstyle{definition}
\newtheorem{definition}{Definition}
\newtheorem{example}{Example}

\theoremstyle{remark}
\newtheorem{remark}{Remark}
\newtheorem*{remark*}{Remark}

\def\eqns#1{\begin{equation*}#1\end{equation*}}
\def\eqnl#1#2{\begin{equation}\label{#1}#2\end{equation}}
\def\eqnsa#1{\begin{subequations}\begin{align*}#1\end{align*}\end{subequations}}
\def\eqnsml#1{\begin{multline*}#1\end{multline*}}

\def\T{\mathrm{T}}

\def\zero{\mathbf{0}}
\def\one{\mathbf{1}}
\def\boC{\mathbf{C}}
\def\boE{\mathbf{E}}
\def\boF{\mathbf{F}}
\def\boI{\mathbf{I}}
\def\bsm{\bm{m}}
\def\boM{\mathbf{M}}
\def\boL{\mathbf{L}}

\def\bsp{\bm{p}}
\def\bsP{\bm{P}}
\def\boX{\mathbf{X}}
\def\boY{\mathbf{Y}}

\def\calA{\mathcal{A}}
\def\calB{\mathcal{B}}
\def\calE{\mathcal{E}}

\def\calL{\mathcal{L}}

\def\calX{\mathcal{X}}
\def\bscalX{\bm{\mathcal{X}}}

\def\bbN{\mathbb{N}}
\def\bbP{\mathbb{P}}
\def\bbR{\mathbb{R}}

\def\d{\mathrm{d}}
\def\s{\mathrm{s}}

\DeclareMathOperator{\Mod}{mod}

\def\AND{\qquad\mbox{ and }\qquad}
\def\defeq{\doteq}
\def\ET{,\;\;}
\def\frall#1{(#1)\qquad}
\def\given{\,|\,}

\def\implies{\Rightarrow}

\def\ind#1{\one_{#1}}
\def\inteq{\stackrel{\smallint}{=}}
\def\scl{\dagger}

\def\sst{\,:\,} 
\def\Ssup#1{\|#1\|}
\def\ssup#1{\|#1\|}
\def\st{\;\,\mbox{s.t.}\;\,}

\def\fpstar{\odot}
\def\vphi{\varphi}